\documentclass[%
reprint,
superscriptaddress,
amsmath,amssymb,
aps,
]{revtex4-1}
\usepackage{graphicx}
\usepackage{dcolumn}
\usepackage{bm}
\usepackage[dvipdfm,colorlinks=true,linkcolor=blue,citecolor=red,urlcolor=blue]{hyperref}
\usepackage{natbib}
\usepackage{mathrsfs,mathtools,color,wasysym,dsfont,here}
\usepackage{amsthm}
\theoremstyle{definition}

\newtheorem{lemm}{LEMMA}[section]
\usepackage{physics}

\begin{document}

\preprint{}

\title{Dissipative quantum Ising chain as a non-Hermitian Ashkin-Teller model}

\author{Naoyuki Shibata}
\email{shibata-naoyuki@g.ecc.u-tokyo.ac.jp}
\affiliation{Department of Physics, Graduate School of Science, the University of Tokyo, 7-3-1 Hongo, Tokyo 113-0033, Japan
}
\author{Hosho Katsura}
\affiliation{Department of Physics, Graduate School of Science, the University of Tokyo, 7-3-1 Hongo, Tokyo 113-0033, Japan
}
\affiliation{Institute for Physics of Intelligence, The University of Tokyo, 7-3-1 Hongo, Tokyo 113-0033, Japan
}

\date{\today}

\begin{abstract}
	We study a quantum Ising chain with tailored bulk dissipation, which can be mapped onto a non-Hermitian Ashkin-Teller model. By exploiting the Kohmoto-den Nijs-Kadanoff transformation, we further map it to a staggered XXZ spin chain with pure-imaginary anisotropy parameters. This allows us to study the eigenstates of the original Liouvillian in great detail. We show that the steady state in each parity sector is a completely mixed state. The uniqueness of each is proved rigorously. We then study the decay modes on the self-dual line corresponding to the uniform XXZ chain and obtain an exact formula for the Liouvillian gap $ g $, the inverse relaxation time, in the thermodynamic limit. The gap $ g $ as a function of dissipation strength $ \Delta $ has a cusp, implying a kind of phase transition for the first decay mode.

\end{abstract}

\maketitle


\section{Introduction}
	Open quantum systems have recently attracted much attention in a variety of fields including condensed matter physics~\cite{Diehl2008,Prosen2008,Diehl2010,Diehl2011,Bardyn2013}, quantum information~\cite{Kraus2008,Kastoryano2011}, quantum computing~\cite{Verstraete2009}, and mathematical physics~\cite{Prosen2013}. The Lindblad equation~\cite{Breuer2002} is a general quantum master equation describing such open quantum systems under Markovian and Completely Positive and Trace Preserving (CPTP) conditions. 
In general, 
the analysis of the Lindblad equation is more challenging than that of the Schr\"{o}dinger equations for closed systems, as one has to deal with the space of opertors rather than the Hilbert space. 
One possible approach is to develop some 
approximate methods such as the perturbative 
\cite{Li2014,Znidaric2015} or numerical ones~\cite{Prosen2009,Znidaric2011}. In particular, due to the recent development of machine learning approaches~\cite{Hartmann2019,Nagy2019,Vicentini2019,Yoshioka2019}, the number of these studies has been increasing. Another way which is complementary to the above methods is to construct exactly solvable models. Although many such models are known in closed many-body systems, very few exact results are available for open many-body systems~\cite{Prosen2011,Karevski2013,Prosen2013,Prosen2014,Medvedyeva2016,Shibata2019}.
	 
The quantum Ising chain is a paradigmatic example of an exactly solvable model for closed systems and also serves as a textbook example of a quantum phase transition~\cite{Sachdev1999,Dutta2015}. It is mapped to a 
free-fermion model, and hence integrable, which allows for explicit computation of various quantities. 
With the recent surge of interest in open quantum systems, a number of studies on dissipative quantum Ising models have been reported recently~\cite{Ates2012,DeLuca2013,Horstmann2013,Vasiloiu2018}. 
However, to the best of our knowledge, exact solutions are available only for models subject to dissipation at the end of the chain~\cite{Prosen_2008}. 
	
	In this paper, we present a dissipative 
dissipative quantum Ising model with bulk dissipation, which is integrable with judiciously chosen dissipators and parameters. 
The Hamiltonian and dissipators of the model conserve parity, and hence it has 
two degenerate non-equilibrium steady states (NESSs). By vectorizing the density matrix, one can identify the Liouvillian of the system with a non-Hermitian analog of the Ashkin-Teller model, the Hermitian counterpart of which was studied in Refs.~\cite{kohmoto1981,alcaraz1988, yamanaka1993phase,alcaraz1984finite,yamanaka1994phase}. 
Due to the $\mathbb{Z}_2 \times \mathbb{Z}_2$ symmetry of the model, the space of states splits into four sectors labeled by the eigenvalues of the parity operators. 
In each sector, the Hamiltonian of the model can be mapped to that of a staggered XXZ chain with pure-imaginary anisotropy parameters. 
This 
enables us to prove that the two NESSs are unique. 
By further exploiting this mapping,
we investigate the Liouvillian gap $ g $, the inverse 
relaxation time, and the corresponding first decay mode on the self-dual line at which the bulk Hamiltonian of the XXZ chain is spatially uniform. Due to the non-local nature of the transformation, the boundary terms in the XXZ chain differs from sector to sector. 
With this in mind, we prove rigorously that the Liouvillian gap in two sectors are exactly $ 4\Delta $. 
Furthermore, we find that the gap in the other two sectors in the thermodynamic limit can be obtained analytically from the Bethe ansatz solution of the model~\cite{Qiao2018}.  
Combining these two results, we derive an explicit formula for the global Liouvillian gap $ g $ as a function of the dissipation strength $ \Delta $, which has a cusp at $ \Delta =1/\sqrt{3} $. 

The rest of the paper is organized as follows. In Sec.~\ref{sec:model}, we give a precise definition of the model and discuss its NESSs. In Sec.~\ref{sec:map_to_staggered_XXZ}, we derive the mapping of the model to the staggered XXZ model in a different way from Ref.~\cite{kohmoto1981}. We also consider the boundary terms carefully and show that that of the XXZ chain differs by sectors. In Sec.~\ref{sec:Proof_of_the_uniqueness_of_NESS}, we show a lemma on the largest imaginary part of eigenvalues of non-Hermitian matrix. With the help of this lemma and the above mapping, we prove that the two NESSs constructed in Sec.~\ref{sec:model} are unique. In Sec.~\ref{sec:Liouvillian_gap}, we discuss the Liouvillian gap of the model on the self-dual line in each of four sectors, and obtain the exact formula of the global Liouvillian gap. In Appendix~\ref{app:parafermion}, we present another formalism of the model, i.e., the $ \mathbb{Z}_4 $ parafermion chain with non-Hermitian interactions.

\section{MODEL AND NESS\lowercase{s}}\label{sec:model}
	Consider the Lindblad equation
	\begin{align}
		\dv{\rho}{t}=\mathcal{L}[\rho]\coloneqq-\mathrm{i}[H,\rho]+\sum_{i}\qty(L_i\rho L_i^\dagger-\dfrac{1}{2}\qty{L_i^\dagger L_i,\rho})\label{eq:Lindblad}
	\end{align}
	for a quantum Ising chain described by the Hamiltonian
	\begin{align}
		H&=-h\sum_{i=1}^{N}\sigma_i^z -J\sum_{i=1}^{N} \sigma_i^x\sigma_{i+1}^x\label{eq:boundary_ham}
	\end{align}
	with the following Lindblad operators
	\begin{align}
		L_i&=\sqrt{\Delta_1}\sigma_i^z\; (i=1,\dots, N),\\ L_{N+i}&=\sqrt{\Delta_2}\sigma_i^x\sigma_{i+1}^x\; (i=1,\dots N)\label{eq:boundary_dis}.
	\end{align}
	Here $ \rho $ is the density matrix, $ N $ is the number of sites, $ \sigma_i^\alpha $ $ (i=1,\dots, N,\; \alpha=x,z) $ are Pauli matrices at site $ i $, and $ \Delta_1,\Delta_2\ge 0 $ are 
dissipation strengths. We impose periodic boundary conditions $ \sigma_{N+1}^x=\sigma_{1}^x $. 
Note that the Lindblad operators take the same form as the local Hamiltonians of $ H $. 
Since we can change the signs of $h$ and $J$ by an appropriate unitary transformation, we  can assume that $ h,J\ge 0 $ without loss of generality.  
	
	The model has a conserved charge, the \textit{parity operator}
	\begin{align}
		P\coloneqq \prod_{i=1}^{N}\sigma_i^z,
	\end{align}
	which satisfies
	\begin{align}
		[H,P]=0,\quad [L_i, P]=0\;\; (i=1,\dots, 2N).
	\end{align}
	Due to the existence of this conserved charge, there are two non-equilibrium steady states (NESSs)
	\begin{align}
		\rho_\pm\coloneqq \dfrac{1\pm P}{2^N},
	\end{align}
	which are eigenstates of the Liouvillian with eigenvalue $ 0 $:
	\begin{align}
		\mathcal{L}[\rho_\pm]=0.\label{eq:NESS_fixed_point_cond}
	\end{align}
	We briefly explain this construction of NESSs according to Ref.~\cite{Shibata2019}. Since all Lindblad operators are Hermitian, there is a trivial NESS, i.e., completely mixed state $ \rho_\mathrm{c}\coloneqq \mathds{1}/2^N $. By noting that $ P\rho_\mathrm{c} \propto P $ also satisfies $ \mathcal{L}[P]=0 $, we see that $ \rho_{\pm} $ obeys Eq.~(\ref{eq:NESS_fixed_point_cond}). We note that the operators $ \rho_{\pm} $ are positive semi-definite, as their eigenvalues are nonnegative. 
	
	\begin{figure}
		\centering
		\includegraphics[width=1.0\linewidth]{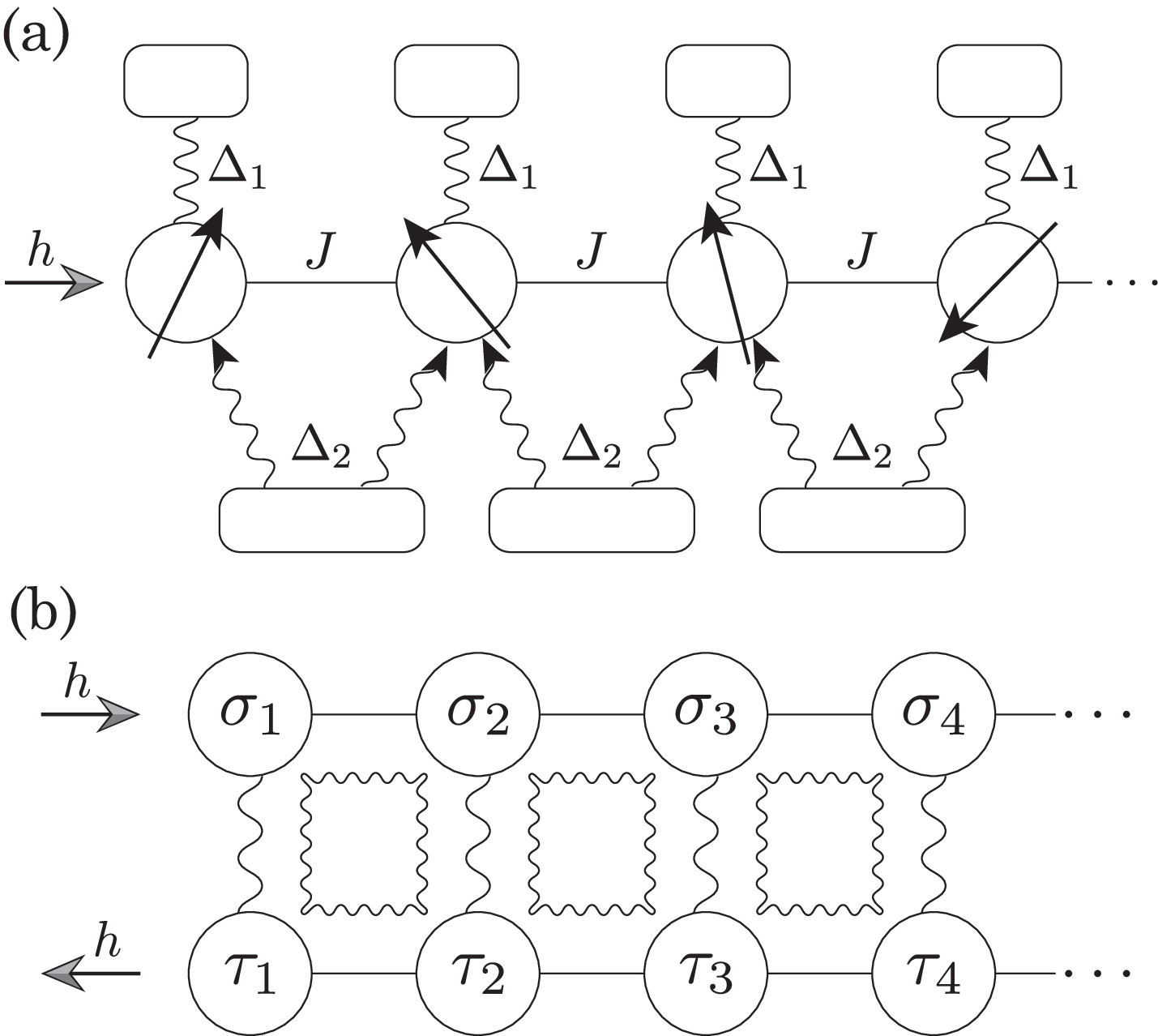}
		\caption{Schematic representations of (a) the set-up and (b) the Liouvillian (times $ \mathrm{i} $) on the $ \mathrm{Ket}\otimes\mathrm{Bra} $ space [see Eqs.~(\ref{eq:Lindbladian_mapped_to_non-Hermitian}) and (\ref{eq:TFIM_ladder})]. Rectangles in (a) denotes reservoirs. Solid and wavy lines in (b) are referred to as Hermitian and non-Hermitian interactions, respectively.}
		\label{fig:schematic_representation_of_set-up}
	\end{figure}
	
\section{Mapping to the non-Hermitian staggered XXZ model}\label{sec:map_to_staggered_XXZ}
	\subsection{Non-Hermitian Ashkin-Teller model and staggered XXZ model}
		
	\begin{figure*}
		\centering
		\includegraphics[width=1.0\linewidth]{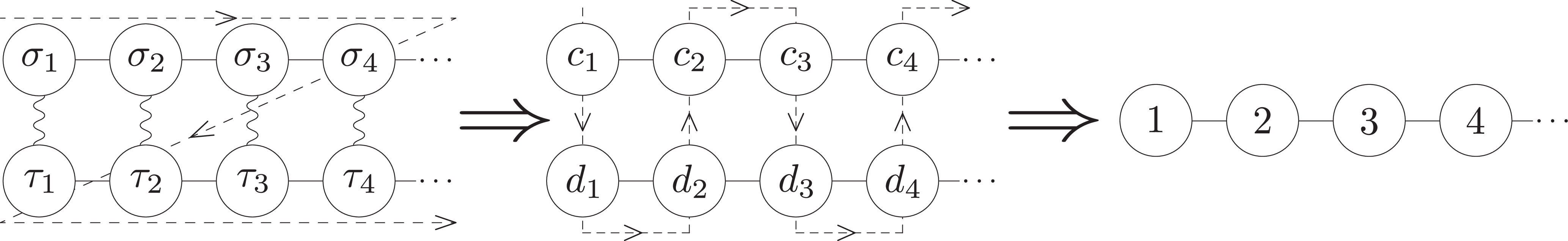}
		\caption{Mapping procedure of our model. The left, middle, and right figures refer to Eqs.~(\ref{eq:TFIM_ladder_bulk}), (\ref{eq:ladder_Majorana}), and (\ref{eq:XXZ_with_pure-imaginary_Delta}), respectively. 
In the left and middle panels, the dashed lines indicate the ordering paths used for the Jordan-Wigner transformation. 
}
		\label{fig:ladderjordan-wignertransformation}
	\end{figure*}
		A $ 2^N\times 2^N $ density matrix $ \rho $ can be thought of as a $ 2^{2N} $-dimensional vector~\cite{Znidaric2014,Znidaric2015,Minganti2018,Shibata2019}. In this sense, we can identify the Liouvillian $ \mathcal{L} $ (times $ \mathrm{i} $) as a non-Hermitian Hamiltonian on a ``$ \mathrm{Ket}\otimes\mathrm{Bra} $ space'' as
		\begin{align}
			\mathrm{i}\mathcal{L}&\cong H\otimes\mathds{1}-\mathds{1}\otimes H^\mathrm{T}\nonumber\\
			&\quad+\mathrm{i}\sum_{i}\qty(L_i\otimes L_i^\ast-\dfrac{1}{2}L_i^\dagger L_i\otimes\mathds{1}-\dfrac{1}{2}\mathds{1}\otimes L_i^\mathrm{T}L_i^\ast),\label{eq:Lindbladian_mapped_to_non-Hermitian}
		\end{align}
		where the Hilbert space of the RHS is the ``$ \mathrm{Ket}\otimes\mathrm{Bra} $ space''. For our model, 
the corresponding non-Hermitian Hamiltonian reads
		\begin{align}
			\mathrm{i}\mathcal{L}+\text{const.}&\cong\mathcal{H}=\mathcal{H}_\mathrm{bulk}+\mathcal{H}_\mathrm{boundary},\label{eq:TFIM_ladder}\\
			\mathcal{H}_\mathrm{bulk}&=-h\sum_{i=1}^{N}\sigma_i^z -J\sum_{i=1}^{N-1} \sigma_i^x\sigma_{i+1}^x\notag\\
			&\hspace{1em}+h\sum_{i=1}^{N}\tau_i^z +J\sum_{i=1}^{N-1} \tau_i^x\tau_{i+1}^x\notag\\
			&\hspace{1em}+\mathrm{i}\Delta_1 \sum_{i=1}^{N} \sigma_i^z \tau_i^z +\mathrm{i}\Delta_2 \sum_{i=1}^{N-1} \sigma_{i}^x\sigma_{i+1}^x\tau_i^x\tau_{i+1}^x,\label{eq:TFIM_ladder_bulk}\\
			\mathcal{H}_\mathrm{boundary}&=J\sigma_{N}^x\sigma_{1}^x+J\tau_{N}^x\tau_{1}^x+\mathrm{i}\Delta_2 \sigma_{N}^x\sigma_{1}^x\tau_{N}^x\tau_{1}^x\label{TFIM_ladder_boundary}
		\end{align}
		where $ \tau_i^\alpha $ ($ \alpha=x,z $) are the Pauli matrices for the $ i $th Bra site (see Fig. \ref{fig:schematic_representation_of_set-up} (b)). First, let us concentrate on the bulk Hamiltonian $ \mathcal{H}_\mathrm{bulk} $. It corresponds to the quantum Ashkin-Teller model~\cite{kohmoto1981,alcaraz1988} with imaginary anisotropy parameters by an appropriate unitary transformation~\footnote{In the case where $ N $ is odd, our model and the quantum Ashkin-Teller model has difference in the boundary terms.}. Furthermore, it is mapped to the staggered XXZ model. It was already mentioned in Refs.~\cite{kohmoto1981,alcaraz1988}, but let us derive it in another way using Majorana fermions. By a Jordan-Wigner transformation
		\begin{align}
			\sigma_i^z&=-\mathrm{i}c_{2i-1}c_{2i},\\
			\sigma_i^x&=\qty(\prod_{j=1}^{i-1}-\mathrm{i}c_{2j-1}c_{2j})c_{2i-1},\\
			\tau_i^z&=-\mathrm{i}d_{2i-1}d_{2i},\\
			\tau_i^x&=\qty(\prod_{j=1}^N-\mathrm{i}c_{2j-1}c_{2j}) \qty(\prod_{k=1}^{i-1}-\mathrm{i}d_{2k-1}d_{2k})d_{2i-1},
		\end{align}
		the model is mapped to the interacting Majorana fermion model with the Hamiltonian (Fig.~\ref{fig:ladderjordan-wignertransformation})
		\begin{align}
			\mathcal{H}_\mathrm{bulk}&= \mathrm{i}h\sum_{i=1}^N c_{2i-1}c_{2i}+\mathrm{i}J\sum_{i=1}^{N-1} c_{2i}c_{2i+1}\notag\\
			&\; -\mathrm{i}h\sum_{i=1}^N d_{2i-1}d_{2i}-\mathrm{i}J\sum_{i=1}^{N-1} d_{2i}d_{2i+1}\notag\\
			&\; -\sum_{i=1}^{N}\mathrm{i}\Delta_1 c_{2i-1}c_{2i}d_{2i-1}d_{2i}\notag\\
			&\, -\sum_{i=1}^{N-1}\mathrm{i}\Delta_2 c_{2i}c_{2i+1}d_{2i}d_{2i+1}.\label{eq:ladder_Majorana}
		\end{align}
		Here, $ c_i $ and $ d_i $ ($ i=1,\dots, 2N $) are Majorana operators obeying the following relations
		\begin{gather}
			c_i^\dagger=c_i,\quad d_i^\dagger=d_i, \\
			\qty{c_i, c_j}=\qty{d_i,d_j}=2\delta_{ij},\quad \qty{c_i,d_j}=0.
		\end{gather}
		Now we apply another Jordan-Wigner transformation 
and rewrite the Majorana fermions as
		\begin{align}
			c_{2i-1}&=(-1)^{i}\qty(\prod_{j=1}^{2i-2}Z_j)X_{2i-1},\\
			d_{2i-1}&=(-1)^{i}\qty(\prod_{j=1}^{2i-2}Z_j)Y_{2i-1},\\
			c_{2i}&=(-1)^{i}\qty(\prod_{j=1}^{2i-1}Z_j)Y_{2i},\\
			d_{2i}&=(-1)^{i}\qty(\prod_{j=1}^{2i-1}Z_j)X_{2i},
		\end{align}
where $ X_i, Y_i $ and $ Z_i $ ($ i=1,\dots 2N $) are Pauli operators at site $ i $. One can express these new Pauli operators as non-local products of the original ones, i.e., $ \sigma_j $ and $ \tau_j $. 
%
Using the new Paulli operators, $ \mathcal{H}_\mathrm{bulk} $ can be recast as
		\begin{align}
			\mathcal{H}_\mathrm{bulk}&=\sum_{i=1}^{N}\qty[h\qty(X_{2i-1} X_{2i}+Y_{2i-1} Y_{2i})+\mathrm{i}\Delta_1Z_{2i-1} Z_{2i}]\notag\\
			&\; +\sum_{i=1}^{N-1}\qty[J\qty(X_{2i}X_{2i+1} +Y_{2i} Y_{2i+1})+\mathrm{i}\Delta_2Z_{2i} Z_{2i+1}],\label{eq:XXZ_with_pure-imaginary_Delta}
		\end{align}
which is nothing but the Hamiltonian of the staggered XXZ model, with the caveat that the two anisotropy parameters are purely imaginary~\footnote{The same model can be constructed from the one-dimensional quantum compass model
			\begin{align}
			H=-\sum_{i=1}^{N/2} J_x\sigma_{2i-1}^x\sigma_{2i}^x -\sum_{i=1}^{N/2-1} J_y\sigma_{2i}^y\sigma_{2i+1}^y
			\end{align}
			with the Lindblad operators
			\begin{align}
			L_{2i-1}=\sqrt{\Delta_1}\sigma_{2i-1}^x\sigma_{2i}^x,\; L_{2i}=\sqrt{\Delta_2}\sigma_{2i}^y \sigma_{2i+1}^y,
			\end{align}
			although there exist many conserved charges in this case, and hence degenerate NESSs.}. 
When $ h=J $ and $ \Delta_1=\Delta_2 $, 
the model reduces to the uniform XXZ chain with pure imaginary anisotropy parameter, which is further investigated in Sec.~\ref{sec:Liouvillian_gap}. From another point of view, the bulk Hamiltonian can be thought of as a $\mathbb{Z}_4$ parafermion chain with non-Hermitian interactions 
(see Appendix~\ref{app:parafermion} for more detail).

	\subsection{Boundary terms}\label{subsec:boundary}
		Here 
we consider the boundary terms in Eqs.~(\ref{eq:boundary_ham}) and (\ref{eq:boundary_dis}), which lead to Eq. (\ref{TFIM_ladder_boundary}). It is useful to define the following operators
	\begin{align}
		P_\sigma\coloneqq \prod_{i=1}^{N} \sigma_{i}^z,\quad P_\tau\coloneqq \prod_{i=1}^{N} \tau_{i}^z.
	\end{align}
	They are conserved charges of $ \mathcal{H} $ satisfying
	\begin{align}
		[\mathcal{H},P_{\sigma/\tau}]=0,\quad [P_\sigma,P_\tau]=0.
	\end{align}
	By the first Jordan-Wigner transformation, 
they can be written in terms of Majorana fermions as
	\begin{align}
		P_\sigma=(-\mathrm{i})^N\prod_{j=1}^{N}c_{2j-1}c_{2j},\quad P_\tau=(-\mathrm{i})^N\prod_{j=1}^{N}d_{2j-1}d_{2j},
	\end{align}
	and the boundary terms are mapped to
	\begin{align}
		\mathcal{H}_\mathrm{boundary}&=-\mathrm{i}J\qty(P_\sigma c_{2N}c_1-P_\tau d_{2N}d_{1})\notag\\
		&\hspace{1em}-\mathrm{i}\Delta_2 P_\sigma P_\tau c_{2N} c_1 d_{2N} d_1.
	\end{align}
	By the second Jordan-Wigner transformation, we have
	\begin{gather}
		P_\sigma=(-1)^N \prod_{j=1}^{2N} Y_j,\quad P_\tau=\prod_{j=1}^{2N} X_j,\\
		\begin{aligned}
			\mathcal{H}_\mathrm{boundary}&=(-1)^N J(P_\tau X_{2N}X_1 +P_\sigma Y_{2N}Y_1)\\
			&\hspace{1em}+\mathrm{i}\Delta_2 P_\sigma P_\tau Z_{2N}Z_1.
		\end{aligned}
	\end{gather}
	We then define a new set of $ \mathbb{Z}_2 $ charges as
	\begin{align}
		Q_Z\coloneqq P_\sigma P_\tau=\prod_{j=1}^{2N}Z_j,\quad Q_X\coloneqq P_\tau =\prod_{j=1}^{2N} X_j,
	\end{align}
in terms of which the 
the boundary terms can be rewritten as
	\begin{align}
		\mathcal{H}_\mathrm{boundary}&=(-1)^N JQ_X(X_{2N}X_1 +Q_Z Y_{2N}Y_1)\notag\\
		&\hspace{1em}+\mathrm{i}\Delta_2 Q_Z Z_{2N}Z_1.
    \label{eq:boundary_XXZ}
	\end{align}
Equation (\ref{eq:boundary_XXZ}) clearly shows that the Hilbert space of the Hamiltonian $ {\cal H} = {\cal H}_{\rm bulk} + {\cal H}_{\rm boundary} $ splits into the following sectors labeled by the eigenvalues of $ Q_Z $ and $ (-1)^N Q_X $: 
	\renewcommand{\labelenumi}{(\roman{enumi})}
	\begin{enumerate}
		\item $ Q_Z=+1,\; (-1)^N Q_X=+1 $: periodic
		
		\item $ Q_Z=+1,\; (-1)^N Q_X=-1 $: anti-periodic
		
		\item $ Q_Z=-1,\; (-1)^N Q_X=+1 $: anti-diagonal twisted~\cite{Batchelor1995,Niekamp2009a}
		
		\item $ Q_Z=-1,\; (-1)^N Q_X=-1 $: anti-periodic and anti-diagonal twisted
	\end{enumerate}
	We now define
	\begin{align}
		\mathcal{H}_\mathrm{boundary}(a, b)&=b J(X_{2N}X_1 +a Y_{2N}Y_1)\notag\\
		&\hspace{1em}+\mathrm{i}\Delta_2 a Z_{2N}Z_1,\label{eq:XXZ_boundary}\\
		\widetilde{\mathcal{H}}(a, b)&=\mathcal{H}_\mathrm{bulk}+\mathcal{H}_\mathrm{boundary}(a,b)
	\end{align}
	where $ a $ and $ b $ are c-numbers, instead of q-numbers, taking $ \pm 1 $. 
One can construct the eigenstates of $ {\cal H} $ from simultaneous eigenstates of $ {\widetilde {\cal H}}(a, \pm 1) $ and $ Q_Z $. This can be seen by noting that if $ \ket{\phi} $ is a simultaneous eigenstate of $ {\widetilde {\cal H}}(a, \pm 1) $ and $ Q_Z $, then $ (\ket{\phi}\pm(-1)^N Q_X\ket{\phi})/2 $ is an eigenstate of $ \mathcal{H} $ with the same eigenvalue.

It is known that $ \widetilde{\mathcal{H}}(a,b) $ for any $ a=\pm 1 $ and $ b=\pm 1 $ is integrable (see, e.g., Ref.~\cite{Niekamp2009a} for more detail). In particular, $ \widetilde{\mathcal{H}}(+1,b) $ has a $ \mathrm{U}(1) $ symmetry, i.e., the $ z $ component of the total spin $ Z^\mathrm{tot}=\sum_i Z_i $ commutes with $ \widetilde{\mathcal{H}}(+1, b) $. Therefore, eigenvalues and eigenvectors in sector (i) and (ii) can be obtained by means of the algebraic Bethe ansatz (ABA). However, for $ \widetilde{\mathcal{H}}(-1,b) $, $ Z^\mathrm{tot} $ is no longer a conserved quantity and the standard ABA fails to work in sector (iii) and (iv). Instead, there is a general method for studying such models without $ \mathrm{U}(1) $ symmetry called the off-diagonal Bethe ansatz (ODBA)~\cite{Niekamp2009a,Cao2013,Wang2015,Qiao2018}.
	
\section{Proof of the uniqueness of the two NESS\lowercase{s}}\label{sec:Proof_of_the_uniqueness_of_NESS}
	
	It can be proved that all eigenvalues of the Liouvillian $ \mathcal{L} $ have non-positive real parts~\cite{Breuer2002,Rivas2011}. Then, the eigenvalue $ 0 $ of $ \mathcal{L} $, whose corresponding eigenstate is a NESS, has the largest \textit{real} part. After the above mapping, the corresponding eigenvalue of $ \mathcal{H} $ has the largest \textit{imaginary} part and its real part is $ 0 $. 
It turns out that such states are \textit{ferromagnetic states}:
	\begin{align}
		\ket{\Uparrow}\coloneqq \ket{\uparrow\dots\uparrow},\quad \ket{\Downarrow}\coloneqq \ket{\downarrow\dots\downarrow},
	\end{align}
	where $ \ket{\uparrow}_i $ (resp. $ \ket{\downarrow}_i $) is the eigenstates of $ Z_i $ with the eigenvalue $ +1 $ (resp. $ -1 $). This follows from the following lemma.
	\begin{lemm}\label{lem:im_bound}
		Let $ \mathcal{H}=T+\mathrm{i}K $, with $ T^\dagger=T $ and $ K^\dagger=K $, be a non-Hermitian matrix, $ \lambda_i $ be eigenvalues of $ \mathcal{H} $, and $ \epsilon_K $ be the largest eigenvalue of $ K $. Then,
		\begin{align}
			\max_i \qty(\Im\lambda_i)\le \epsilon_K.\label{eq:non-Hermitian_spectrum_bound}
		\end{align}
	\end{lemm}
	\begin{proof}
		Let $ \ket{\psi_i} $ be a right-eigenstate of $ \mathcal{H} $ with eigenvalue $ \lambda_i $. We can assume that $ \ket{\psi_i} $ is normalized as $ \bra{\psi_i}\ket{\psi_i}=1 $. Then we have
		\begin{align}
			\lambda_i=\bra{\psi_i}\mathcal{H}\ket{\psi_i}=\bra{\psi_i}T\ket{\psi_i}+\mathrm{i}\bra{\psi_i}K\ket{\psi_i}
		\end{align}
		Since $ T $ and $ K $ are Hermitian, $ \bra{\psi_i}T\ket{\psi_i} $ and $ \bra{\psi_i}K\ket{\psi_i} $ are real. Thus we find that
		\begin{align}
			\Im \lambda_i=\bra{\psi_i}K\ket{\psi_i}\le \epsilon_K. 
		\end{align}
		This holds for all $ i $, so Eq. (\ref{eq:non-Hermitian_spectrum_bound}) also holds.
	\end{proof}
	For our Hamiltonian $ \mathcal{H}=\mathcal{H}_\mathrm{bulk}+\mathcal{H}_\mathrm{boundary} $,
	\begin{align}
		K&=(\mathcal{H}-\mathcal{H}^\dagger)/(2\mathrm{i})\\
		&=\Delta_1\sum_{i=1}^N Z_{2i-1}Z_{2i}+\Delta_2\sum_{i=1}^{N-1} Z_{2i}Z_{2i+1}\notag\\
		&\hspace{1em} +\Delta_2 Q_Z Z_{2N} Z_1
	\end{align}
	is already diagonalized. The largest eigenvalue is $ N(\Delta_1+\Delta_2) $ and the corresponding (unique) eigenstates are $ \ket{\Uparrow} $ and $ \ket{\Downarrow} $. Moreover,
	\begin{align}
		T\ket{\Uparrow}=T\ket{\Downarrow}=0,
	\end{align}
	where
	\begin{align}
		T&=(\mathcal{H}+\mathcal{H}^\dagger)/2\\
		&=\sum_{i=1}^{N}h\qty(X_{2i-1} X_{2i}+Y_{2i-1} Y_{2i})\notag\\
		&\hspace{1em}+\sum_{i=1}^{N-1}J\qty(X_{2i}X_{2i+1} +Y_{2i} Y_{2i+1})\notag\\
		&\hspace{1em} +(-1)^N JQ_X(X_{2N}X_1 +Q_Z Y_{2N}Y_1)
	\end{align}
	Therefore, these ferromagnetic states correspond to the two unique NESSs of the original model. One can easily see that the superposition $ (\ket{\Uparrow}+(-1)^N\ket{\Downarrow})/\sqrt{2} $ (resp. $ (\ket{\Uparrow}-(-1)^N\ket{\Downarrow})/\sqrt{2} $) lives in sector (i) (resp. sector (ii)).
	
\section{The Liouvillian gap on the self-dual line}\label{sec:Liouvillian_gap}
In this section, we focus on our model on the \textit{self-dual line}~\cite{kohmoto1981}
	\begin{align}
		h=J,\quad \Delta_1=\Delta_2=\Delta,
	\end{align}
	on which $ \mathcal{H} $ is invariant under the duality transformation~\cite{DiFrancesco1997}
	\begin{align}
		\sigma_i^z=\widetilde{\sigma}_i^x\widetilde{\sigma}_{i+1}^x,\; \tau_i^z=\widetilde{\tau}_i^x\widetilde{\tau}_{i+1}^x,\; \sigma_i^x=\prod_{k=1}^{i}\widetilde{\sigma}_k^z,\; \tau_i^x=\prod_{k=1}^{i}\widetilde{\tau}_k^z 
	\end{align}
	up to boundary terms. 
We investigate how the Liouvillian gap, i.e., the inverse relaxation time, and the first decay mode behave on the self-dual line. Since the overall scale is not important for the analysis, we set $ h=J=1 $ in the following. Let eigenvalues of the Liouvillian $ \mathcal{L} $ be $ \Lambda_i(\mathcal{L}) $. A Liouvillian gap $ g $ is defined as
	\begin{equation}
		g\coloneqq -\max_{\substack{i\\ \Re[\Lambda_i(\mathcal{L})]\ne 0}}\Re[\Lambda_i(\mathcal{L})],\label{eq:def_of_Liouvillian_gap}
	\end{equation}
	hence, the inverse of the relaxation time. It is clear from Eq.~(\ref{eq:TFIM_ladder}) that the Liouvillian gap corresponds to the gap between the first and second largest \textit{imaginary} parts of eigenvalues of $ \mathcal{H} $.
	
	As we have seen in Sec.~\ref{subsec:boundary} and \ref{sec:Proof_of_the_uniqueness_of_NESS}, the Hilbert space is divided into four sectors (i), (ii), (iii), and (iv), and the two ferromagnetic states which correspond to NESSs live in sector (i) and (ii). We now define the \textit{local} Liouvillian gap $ g_\mathrm{(i)} $, $ g_\mathrm{(ii)} $, $ g_\mathrm{(iii)} $, and $ g_\mathrm{(iv)} $ as follows: $ g_\mathrm{(i)} $ and $ g_\mathrm{(ii)} $ are defined as the gap between the first and second largest imaginary part of the eigenvalues of $ \widetilde{\mathcal{H}}(+1,\pm 1) $. Note that the first largest one is $ 2N\Delta $ as we have proved in Sec.~\ref{sec:Proof_of_the_uniqueness_of_NESS}. The other two gaps, $ g_\mathrm{(iii)} $ and $ g_\mathrm{(iv)} $, are defined as the difference between $ 2N\Delta $ and the largest imaginary part of the eigenvalues of $ \widetilde{\mathcal{H}}(-1,\pm 1) $. The \textit{global} Liouvillian gap $ g $ is then obtained as
	\begin{align}
		g=\min\qty(g_\mathrm{(i)},\; g_\mathrm{(ii)},\; g_\mathrm{(iii)},\; g_\mathrm{(iv)}).
	\end{align}
	The main result of this section is that we the exact formula for $ g $ as a function of dissipation strength $ \Delta $ in the thermodynamic limit is obtained as 
	\begin{align}
		g(\Delta)=
		\begin{dcases}
			4\Delta&\qty(0<\Delta\le \dfrac{1}{\sqrt{3}})\\
			2\sqrt{\Delta^2+1}&\qty(\dfrac{1}{\sqrt{3}}\le \Delta)
		\end{dcases},
    \label{eq:global_g}
	\end{align}
(see also Fig.~\ref{fig:liouvillian_gap}). 
	The result implies a kind of phase transition for the first decay mode. In fact, the first decay mode lives in sector (i) and (ii) in the weak dissipation region $ 0<\Delta\le 1/\sqrt{3} $, while it lives in sector (iii) and (iv) in the strong dissipation region $ \Delta>1/\sqrt{3} $, which we discuss in the following subsections. 
	\begin{figure}
		\centering
		\includegraphics[width=1.0\linewidth]{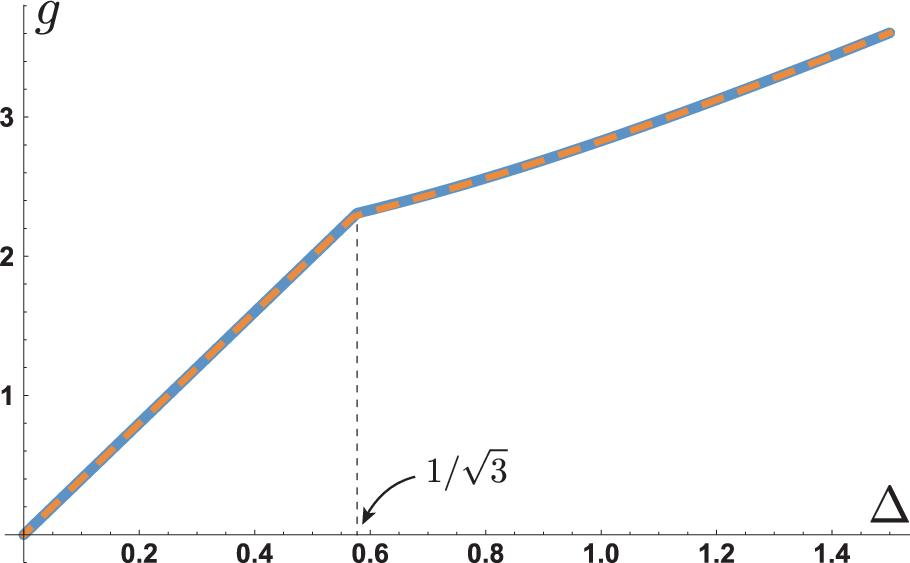}
		\caption{(Color online) Liouvillian gap $ g $ as a function of $ \Delta $. A blue solid line denotes the exact result Eq.~(\ref{eq:global_g}) in the thermodynamic limit possessing a cusp at $ \Delta=1/\sqrt{3} $, while a orange dashed line is the numerical one with $ 2N=10 $.}
		\label{fig:liouvillian_gap}
	\end{figure}
	
	\subsection{In sector (i) and (ii)}
		First, let us consider the lower bound of $ g_\text{(i)} $. In Sec. \ref{sec:Proof_of_the_uniqueness_of_NESS}, we have shown that the largest imaginary part of $ \mathcal{H} $, and also of $ \widetilde{\mathcal{H}}(+1,+1) $, is $ 2N\Delta $ and corresponding eigenstates are ferromagnetic states $ \ket{\Uparrow} $ and $ \ket{\Downarrow} $. Then, the state which has the second largest imaginary part of $ \widetilde{\mathcal{H}}(+1,+1) $ must have at least one up-spin and at least one down-spin. Therefore, there are at least two kinks, which together with Lemma \ref{lem:im_bound} implies that the second largest imaginary part is less than or equal to $ (2N-4)\Delta $. In other words, $ g_\mathrm{(i)}\ge 4\Delta $ holds. Next, let us explicitly construct an eigenstate whose imaginary part of the eigenvalue is $ (2N-4)\Delta $. We denote by $ \ket{i,j} $ the normalized state which has two down-spins at site $ i $ and $ j $, and $ (2N-2) $ up-spins at the rest of chain. Then, one can verify that
		\begin{align}
			\ket{\chi^\text{(i)}}\coloneqq\dfrac{1}{\sqrt{2N}}\sum_{i=1}^{2N}(-1)^{i-1}\ket{i,i+1} \pmod{2N}
		\end{align}
		is an eigenstate of $ \widetilde{\mathcal{H}}(+1,+1) $ with eigenvalue $ (2N-4)\Delta $. The state $ \ket{\chi^\text{(i)}} $ is known as a singular state in the context of the Heisenberg chain~\cite{Avdeev1986,Essler1992,Noh2000,Nepomechie2013}. Therefore, we have proved that $ g_\mathrm{(i)}=4\Delta $ holds and the first decay mode in sector (i) is
		\begin{align}
			\dfrac{(1+(-1)^N Q_X)}{\sqrt{2}}\ket{\chi^\text{(i)}}.
		\end{align}
		
		The Liouvillian gap in sector (ii) $ g_\text{(ii)} $ and the first decay mode can also be obtained in a similar way. One can see that
		\begin{align}
			\ket{\chi^\text{(ii)}}\coloneqq\dfrac{1}{\sqrt{2N}}\sum_{i=1}^{2N-1}(-1)^{i-1}\ket{i,i+1}+\ket{2N,1}
		\end{align}
		is an eigenstate of $ \widetilde{\mathcal{H}}(+1,-1) $ with eigenvalue $ (2N-4)\Delta $. One can also prove that $ g_\text{(ii)}\ge 4\Delta $ using Lemma \ref{lem:im_bound}, and therefore, we have proved rigorously that
		\begin{align}
			g_\text{(ii)}=4\Delta,\label{eq:g_in_sector_ii}
		\end{align}
		and the corresponding first decay mode is
		\begin{align}
			\dfrac{(1-(-1)^N Q_X)}{\sqrt{2}}\ket{\chi^\text{(ii)}}.
		\end{align}
		
	\subsection{In sector (iii) and (iv)}
		\begin{figure*}
			\centering
			\includegraphics[width=1.0\linewidth]{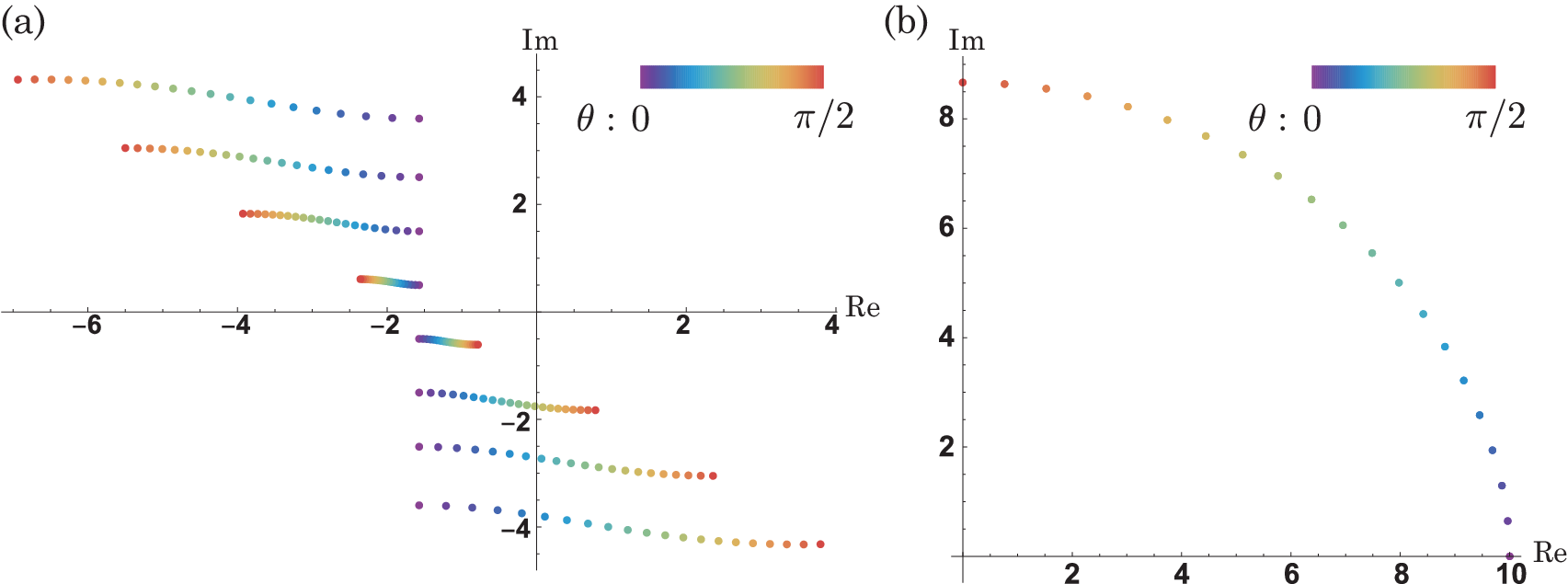}
			\caption{(Color online) Numerical results of (a) Bethe roots and (b) corresponding eigenvalues of $ \sum_{i=1}^{2N}(X_{i}X_{i+1}+Y_{i}Y_{i+1}+\Delta \mathrm{e}^{i\theta}Z_{i}Z_{i+1}) $ under the anti-diagonal twisted boundary condition $ A_{2N+1}=X_1 A_1 X_1 $ ($ A=X,Y,Z $), where $ 2N=8,\; \Delta=\cosh(1.0) $. The results show how the set of Bethe roots and the corresponding eigenvalue move in the complex plane  
as the anisotropy parameter $ \Delta \mathrm{e}^{i\theta} $ varies from $ \theta=0 $ (corresponding to the Hermitian XXZ) to $ \theta=\pi/2 $ (corresponding to $ \widetilde{\mathcal{H}}(-1,+1) $). 
}
			\label{fig:ODBA_analytic_continuation}
		\end{figure*}
		It is known that $ \widetilde{\mathcal{H}}(-1,+1) $ and $ \widetilde{\mathcal{H}}(-1,-1) $ give the same spectrum~\cite{Alcaraz1988a}, which leads to $ g_\mathrm{(iii)}=g_\mathrm{(iv)} $. 
Thus, it suffices to consider only $ \widetilde{\mathcal{H}}(-1,+1) $. In Ref.~\cite{Qiao2018}, the energy of the ferromagnetic XXZ model under this anti-diagonal twisted boundary condition has been studied by the ODBA method. As a result, the authors obtained inhomogeneous Bethe ansatz equations and the formula for the energy $ E $ as
		\begin{widetext}
			\begin{gather}
				\mathrm{e}^{\mathrm{i}u_j}\prod_{l=1}^{2N}\dfrac{\sin(u_j-u_l+\mathrm{i}\eta)}{\sin(u_j+\mathrm{i}\eta/2)}=\mathrm{e}^{-\mathrm{i}u_j}\prod_{l=1}^{2N}\dfrac{\sin(u_j-u_l-\mathrm{i}\eta)}{\sin(u_j-\mathrm{i}\eta/2)}+2\mathrm{i}\mathrm{e}^{-N\eta}\sin(u_j-\sum_{l=1}^{2N}u_l)\quad (j=1,\dots, 2N),\\
				E=-2\mathrm{i}\sinh\eta \sum_{j=1}^{2N} \qty[\cot(u_j+\dfrac{\mathrm{i}\eta}{2})-\cot(u_j-\dfrac{\mathrm{i}\eta}{2})]+2N\sinh\eta+2\sinh\eta,
			\end{gather}
		\end{widetext}
		where $\{ u_j \}$ are the Bethe roots and $ \cosh\eta $ corresponds to the anisotropy parameter.  It is important to note that the ODBA method is also applicable to complex anisotropy parameter, and hence, complex $ \eta $. 
The special case $ \cosh\eta=\mathrm{i}\Delta $ is particularly relevant to the analysis of the gap, as it corresponds to $ \widetilde{\mathcal{H}}(-1,+1) $. 
We confirmed numerically the following two facts:
		\renewcommand{\labelenumi}{\arabic{enumi}.}
		\begin{enumerate}
			\item The Bethe roots for the eigenvalue of $ \widetilde{\mathcal{H}}(-1,+1) $ with the largest imaginary part is obtained by analytic continuation of those for 
the eigenstate of the Hermitian XXZ model with the largest eigenvalue in absolute value (Fig.~\ref{fig:ODBA_analytic_continuation} (b)). 
		 	
			\item The string hypothesis (3.1) of Ref.~\cite{Qiao2018} may also be valid 
for complex anisotropy (Fig.~\ref{fig:ODBA_analytic_continuation} (a)). In particular, the corresponding Bethe roots for $ \widetilde{\mathcal{H}}(-1,+1) $ are obtained by analytic continuation as
			\begin{widetext}
				\begin{align}
				u_j=-\dfrac{\pi}{2}+\qty(\dfrac{2N+1}{2}-j)\mathrm{i}\qty[\ln(\Delta+\sqrt{\Delta^2+1})+\dfrac{\mathrm{i}\pi}{2}]+o(N)\quad (j=1,\dots, 2N).
				\end{align}	
			\end{widetext}
			
		\end{enumerate}
		Then, by using this string hypothesis for the solution, $ g_\mathrm{(iii)} $ in the thermodynamic limit can be obtained as
		\begin{align}
			g_\mathrm{(iii)}=2\sqrt{\Delta^2+1}.\label{eq:g_in_sector_iii}
		\end{align}
		As a result, we obtain the explicit formula for the global Liouvillian gap Eq.~(\ref{eq:global_g}). We have confirmed that this analytical result agrees extremely well with the numerical result  obtained by exact diagonalization for $ 2N=10 $ (see Fig.~\ref{fig:liouvillian_gap}).
		 
		One can check the validity of Eq.~(\ref{eq:g_in_sector_iii}) by considering the Ising limit $ \Delta\to\infty $. 
For notational convenience, we rescale the Hamiltonian 
		\begin{equation}
			\widetilde{\mathcal{H}}(-1,+1) \to 
			\dfrac{1}{\Delta}\widetilde{\mathcal{H}}(-1,+1) = H_0 + V
		\end{equation}
with
		\begin{equation}
			H_0 = \mathrm{i} \sum_{i=1}^{2N} Z_{i}Z_{i+1}
		\end{equation}	
and
		\begin{equation}
			V =  \dfrac{2}{\Delta} \sum_{i=1}^{2N} (S_i^+ S_{i+1}^- +S_i^- S_{i+1}^+),
		\end{equation}
where $ S^\pm_i=(X_i\pm \mathrm{i} Y_i)/2 $ and the anti-diagonal twisted boundary conditions are imposed, i.e., $Z_{2N+1}=-Z_1,\; S_{2N+1}^\pm=S_1^\mp$. 
In the Ising limit, one can first analyze $H_0$, and then treat the remaining term $V$ as a perturbation. 
%
It is seen by inspection that $ (2N-2)\mathrm{i} $ is the eigenvalue of $H_0$ with the largest imaginary part. The corresponding eigenstates with $ Q_Z=-1 $ are the following $ 2N $ states
		\begin{align}
			&\ket{\uparrow \downarrow \dots \downarrow},\; \ket{\uparrow \uparrow \uparrow \downarrow \dots \downarrow},\dots, \ket{\uparrow\dots\uparrow \downarrow},\notag\\
			&\ket{\downarrow \uparrow \dots \uparrow},\; \ket{\downarrow \downarrow \downarrow \uparrow \dots \uparrow},\dots, \ket{\downarrow\dots\downarrow \uparrow}\label{eq:Ising-limit_gs}.
		\end{align}
Each state has two kinks (domain walls), one of which is between sites $2N$ and $1$, and the other of which is between sites $ 2j-1 $ and $ 2j $ ($ j=1,\dots, N $). 
Let $ \mathcal{P} $ be a projection operator onto the subspace spanned by the states in Eq.~(\ref{eq:Ising-limit_gs}). 
We also define 
		\begin{align}
			R\coloneqq \sum_{\substack{n\\ \mathcal{P}\ket{\psi_n}=0}} \dfrac{\ket{\psi_n}\bra{\psi_n}}{(2N-2)\mathrm{i}-E_n},
		\end{align}
		where $ \ket{\psi_n} $ are other eigenstates of the non-perturbed Hamiltonian with eigenvalue $ E_n $. 
Then the effective Hamiltonian from the second order perturbation reads
		\begin{align}
			{\cal H}_{\rm eff} = \mathcal{P}\qty( V+ VRV ) \mathcal{P}.
		\end{align}
		One can see that the first order perturbation term vanishes, and the second order one leads to the constant shift of the eigenvalue by
		\begin{align}
			\dfrac{4}{\Delta^2}\dfrac{1}{(2N-2)\mathrm{i}-(2N-6)\mathrm{i}}=-\dfrac{\mathrm{i}}{\Delta^2}
		\end{align}
		without lifting the degeneracy. Putting all this together, we have
		\begin{align}
			g_\mathrm{(iii)}&=2N\Delta-\qty[(2N-2)\Delta-\dfrac{1}{\Delta}]+\order{\dfrac{1}{\Delta^2}}\\
			&=2\Delta+\dfrac{1}{\Delta}+\order{\dfrac{1}{\Delta^2}},
		\end{align}
		which is consistent with Eq.~(\ref{eq:g_in_sector_iii}).
		 
\section{Summary}
	We have studied a quantum Ising chain with bulk dissipation. By vectorizing the density matrix, we showed that the Liouvillian of the model can be thought of as a non-Hermitian Ashkin-Teller model. 
Then using the Kohmoto-den Nijs-Kadanoff transformation, we further mapped it to 
a staggered XXZ model with pure-imaginary anisotropy parameters. 
This mapping has enabled us to prove that the two NESSs are unique. Furthermore, we obtained the exact results for the Liouvillian gap on the self-dual line corresponding to the uniform XXZ model, which shows an excellent agreement with the numerical results even for small system sizes. 
Though we mostly focused on the self-dual line, it would be interesting to explore the first decay mode in the whole parameter region of $ h,J,\Delta_1 $, and $ \Delta_2 $. Whether it also undergoes the phase transition remains an interesting question for future research.

\begin{acknowledgments}
	The authors thank Nobuyuki Yoshioka and Marko \v{Z}nidari\v{c} for fruitful discussions. H.K. was supported in part by JSPS KAKENHI Grant No. JP18H04478 and No. JP18K03445. N.S. acknowledges support of the Materials Education program for the future leaders in Research, Industry, and Technology (MERIT). 
\end{acknowledgments}

\appendix

\section{Interacting Parafermion Representation}\label{app:parafermion}
	Our model has another representation using $ \mathbb{Z}_4 $ parafermions~\cite{Fendley2012,Moran2017,Calzona2018,Chew2018}. Neglecting the boundary terms, we start from Eq.~(\ref{eq:TFIM_ladder_bulk}). First let us define the following two operators:
	\begin{align}
		\alpha_j&=\dfrac{1+\mathrm{i}}{2}\sigma_j^x+\dfrac{1-\mathrm{i}}{2}(-1)^{j-1}\tau_j^x, \\
		\beta_j&=\dfrac{1}{2}(\sigma_j^z-\tau_j^z)+\dfrac{\mathrm{i}}{2}(-1)^{j}(\sigma_j^x \tau_j^y+\sigma_j^y \tau_j^x),
	\end{align}
	which satisfy
	\begin{align}
		\alpha_j^4=\beta_j^4=1,\quad \alpha_i \beta_j=\mathrm{i}^{\delta_{ij}} \beta_j\alpha_i.
	\end{align}
	One can verify that
	\begin{align}
		\alpha_{i}^\dagger\alpha_{i+1}+\alpha_{i}\alpha_{i+1}^\dagger&=\sigma_{i}^x\sigma_{i+1}^x-\tau_{i}^x\tau_{i+1}^x, \\
		\beta_i+\beta_i^\dagger&=\sigma_{i}^z-\tau_{i}^z. 
	\end{align}
	Then, Eq.~(\ref{eq:TFIM_ladder_bulk}) is rewritten in terms of $\alpha_i$ and $\beta_i$ as
	\begin{align}
		\mathcal{H}_\mathrm{bulk}&=-h\sum_{i=1}^{N}(\beta_i+\beta_i^\dagger)-J\sum_{i=1}^{N-1}(\alpha_i^\dagger\alpha_{i+1}+\alpha_{i}\alpha_{i+1}^\dagger)\notag\\
		&\hspace{1em}-\mathrm{i}\Delta_1 \sum_{i=1}^N\dfrac{\beta_i^2+{\beta_i^\dagger}^2}{2}-\mathrm{i}\Delta_2 \sum_{i=1}^{N-1}\dfrac{(\alpha_i^\dagger\alpha_{i+1})^2+(\alpha_{i}\alpha_{i+1}^\dagger)^2}{2}.
	\end{align}
	Next, we define $ \mathbb{Z}_4 $ parafermion operators as follows:
	\begin{align}
		\gamma_{2i-1}&=\qty(\prod_{j=1}^{i-1}\beta_j)\alpha_i, \\
		\gamma_{2i}&=\mathrm{e}^{\mathrm{i}\pi/4} \qty(\prod_{j=1}^{i-1}\beta_j)\alpha_i\beta_i=\mathrm{e}^{\mathrm{i}3\pi/4}\qty(\prod_{j=1}^{i}\beta_j)\alpha_i.
	\end{align}
	They satisfy
	\begin{align}
		\gamma_i^4=1,\; \gamma_i^\dagger=\gamma_i^3,\; \gamma_i \gamma_j=\mathrm{i}\gamma_j\gamma_i\;(i<j).
	\end{align}
	By noting that
	\begin{align}
		\beta_i&=\mathrm{e}^{-\mathrm{i}\pi/4}\gamma_{2i-1}^\dagger \gamma_{2i}, \\
		\alpha_{i}^\dagger\alpha_{i+1}&=\mathrm{e}^{\mathrm{i}3\pi/4}\gamma_{2i}^\dagger \gamma_{2i+1},
	\end{align}
	we obtain
	\begin{align}
		\mathcal{H}_\mathrm{bulk}=&-h\sum_{i=1}^N \qty(\mathrm{e}^{-\mathrm{i}\pi/4}\gamma_{2i-1}^\dagger \gamma_{2i}+\mathrm{e}^{\mathrm{i}\pi/4}\gamma_{2i}^\dagger \gamma_{2i-1})\notag\\
		&-J\sum_{i=1}^{N-1} \qty(\mathrm{e}^{\mathrm{i}3\pi/4}\gamma_{2i}^\dagger \gamma_{2i+1}+\mathrm{e}^{-\mathrm{i}3\pi/4}\gamma_{2i+1}^\dagger \gamma_{2i})\notag\\
		&-\mathrm{i}\Delta_1\sum_{i=1}^N \gamma_{2i-1}^2\gamma_{2i}^2-\mathrm{i}\Delta_2 \sum_{i=1}^{N-1} \gamma_{2i}^2\gamma_{2i+1}^2.
	\end{align}
   The first and second terms are quadratic, while the third and fourth terms are quartic in parafermions. Therefore, our model can be thought of as a $ \mathbb{Z}_4 $ parafermion chain with non-Hermitian quartic interactions (up to boundary terms).

\bibliography{memo}

\end{document}